\documentclass[12pt]{iopart}

\usepackage{iopams}
\usepackage{amsthm}
\usepackage{hyperref}

\newcommand{\sign}{\mathrm{sign}}

\newtheorem{theorem}{Theorem}
\newtheorem{lemma}{Lemma}
\newtheorem{proposition}{Proposition}

\newcommand{\affiliation}[1]{\address{#1}}
\newcommand{\keywords}[1]{\vspace{10mm}\noindent\textbf{Keywords:} #1}
\newcommand{\openone}{\mbox{{\small$\mathrm{1}$}$\!\!\mathrm{1}$}}
\newcommand{\text}[1]{\mathrm{#1}}
\newcommand{\substack}[1]{\tiny\begin{array}{c}#1\end{array}}

\begin{document}

\title{The propagator of the attractive delta-Bose gas in one dimension}
\date{\today}
\author{Sylvain Prolhac\footnote{prolhac@ma.tum.de}{} and Herbert Spohn\footnote{spohn@ma.tum.de}{}}
\affiliation{Zentrum Mathematik and Physik Department,\\
Technische Universit\"at M\"unchen,\\
D-85747 Garching, Germany}

\begin{abstract}
We consider the quantum $\delta$-Bose gas on the infinite line. For repulsive interactions, Tracy and Widom have obtained an exact formula for the quantum propagator. In our contribution we explicitly perform its analytic continuation to attractive interactions. We also study the connection to the expansion of the propagator in terms of the Bethe ansatz eigenfunctions. Thereby we provide an independent proof of their completeness.
\end{abstract}


\maketitle


\section{Introduction}
\label{Section introduction}
\setcounter{equation}{0}
Quantum particles on the real line interacting through a $\delta$-potential are governed by the Hamiltonian
\begin{equation}
\label{H kappa}
H_{\kappa}=-\sum_{j=1}^{n}\frac{\partial^{2}}{\partial x_{j}^{2}}-2\kappa\sum_{j<k}^{n}\delta(x_{j}-x_{k})\;.
\end{equation}
The number of particles, $n$, is fixed throughout and $x=(x_{1},\ldots,x_{n})$ denotes the positions of the particles. We will restrict ourselves to the bosonic subspace of symmetric wave functions. Eq. (\ref{H kappa}) is the Hamiltonian of a quantum many-body system which can be analyzed through the Bethe ansatz. The repulsive interaction, $\kappa<0$, has been studied in great detail and we refer to \cite{LL63.1,L63.1,YY69.1,G71.1,G71.2,G71.3,G83.1}. The attractive case, $\kappa>0$, has received less attention. One reason is that the structure of the Bethe equations is more complicated. On top, physical applications are not obviously in reach. In the recent years, there has been renewed interest. We have now available a detailed study of the eigenfunctions \cite{MG64.1,KK88.1,CC07.2,D10.2} and, as argued by Calabrese and Caux \cite{CC07.1}, applications to real materials are in sight. A further motivation comes from the one-dimensional Kadar-Parisi-Zhang (KPZ) equation \cite{KPZ86.1}. Its replica solution is given in terms of the propagator of the attractive $\delta$-Bose gas \cite{K87.1,BO90.1} which can be used to obtain exact solutions for some special initial conditions \cite{DK10.1,CLDR10.1,D10.1,D10.2,PS11.1,PS11.2,PS11.3,CLD11.1,IS11.1,CQ11.1,C11.1}.

In the KPZ context, and also in other cases, one is actually interested in the quantum propagator $\langle x|\rme^{-tH_{\kappa}}|y\rangle$, $t\geq0$. In principle, $\rme^{-tH_{\kappa}}$ can be expanded in a sum (integral) over eigenfunctions. But one might hope to have at disposal more concise expressions for the propagator. In the repulsive case, Tracy and Widom \cite{TW08.2} carried out such a program. The resulting expression we refer to as TW formula, which will be discussed below, including its relation to the expansion in eigenfunctions. A natural issue is to extend such a program to the attractive case, which is the topic of our contribution.

By symmetry the propagator $\langle x|\rme^{-tH_{\kappa}}|y\rangle$ can be restricted to the domain $\Lambda=\{x|x_{1}\leq\ldots\leq x_{n}\}\subset\mathbb{R}^{n}$. Using the Bose symmetry, $H_{\kappa}$ of (\ref{H kappa}) is then defined by
\begin{equation}
H_{\kappa}\psi(x_{1},\ldots,x_{n})=-\sum_{j=1}^{n}\frac{\partial^{2}}{\partial x_{j}^{2}}\psi(x_{1},\ldots,x_{n})\;,\quad x\in\Lambda^{\circ}\;,
\end{equation}
with the boundary conditions
\begin{equation}
\left(\frac{\partial}{\partial x_{j+1}}-\frac{\partial}{\partial x_{j}}+\kappa\right)\psi(x_{1},\ldots,x_{n})_{\big|x_{j+1}=x_{j}}=0\;,
\end{equation}
where the limit $x_{j+1}=x_{j}$ is taken from the interior, $\Lambda^{\circ}$, of $\Lambda$. The Hamiltonian $H_{\kappa}$ is a self-adjoint operator and $\langle x|\rme^{-tH_{\kappa}}|y\rangle$ is continuous in $x,y\in\Lambda$. In particular,
\begin{equation}
\lim_{t\to0}\langle x|\rme^{-tH_{\kappa}}|y\rangle=\frac{1}{n!}\prod_{j=1}^{n}\delta(x_{j}-y_{j})\;,\quad x,y\in\Lambda\;.
\end{equation}
Throughout the paper $x,y\in\Lambda$, hence the position of the particles are ordered increasingly. As will be proved in \ref{Appendix analyticity propagator}, $\langle x|\rme^{-tH_{\kappa}}|y\rangle$ is analytic in $\kappa$ for otherwise fixed arguments. Thus our strategy will be to analytically extend the TW formula, valid for $\kappa\leq0$, to $\kappa>0$. As written, the TW formula becomes singular at $\kappa=0$. Therefore the main task is to understand the structure of the analytic continuation in $\kappa$. As a result, we will arrive at various formulas for the propagator. One formula will be just the expansion in Bethe ansatz eigenfunctions, which thus implies their completeness. 

The issue of completeness for the attractive $\delta$-Bose gas on the line has been studied before. In his thesis, Stephen Oxford \cite{Ox79} proves completeness of the generalized eigenfunctions defined as bounded Bethe ansatz eigenfunctions. He uses functional analytic methods to construct the Hilbert space isometry from the generalized eigenfunctions and thereby the spectral representation of $H_\kappa$. A similar strategy is used by Babbitt and Thomas \cite{BT77} for the ground state representation of the ferromagnetic Heisenberg model on $\mathbb{Z}$. Heckman and Opdam \cite{HO97} exploit the fact that the $\delta$-Bose gas turns up in the representation theory of graded Hecke algebras. (We are grateful to Bal\'{a}zs Pozsgay for pointing out this reference.) They have results for the case when the interaction strength is allowed to be pair dependent. But only for $H_\kappa$ their expression simplifies and they arrive at a Plancherel formula, which is the completeness relation.

For the system on the line, studied here, the set of admissible wave numbers is known explicitly. For a bounded system, in particular with periodic boundary conditions, the discrete set of wave numbers are the solutions to the Bethe equations, a coupled system of $n$ transcendental equations. Completeness becomes more difficult to establish and to our knowledge only for the repulsive case 
a completeness proof is available \cite{D93.1}.

The article is organized as follows. In Section \ref{Section repulsive}, we recall the Tracy and Widom formula for the propagator in the repulsive case, and rewrite it in terms of Bethe eigenstates. In Section \ref{Section attractive}, we summarize our main results on the propagator with attractive interactions. These results are proved in Section \ref{Section analytic continuation}, by performing explicitly the analytic continuation to $\kappa>0$. A further rewriting represents the propagator in terms of the known Bethe eigenstates. The special case of the propagator with all particles starting and ending at $0$ is handled in Section \ref{Section x=y=0}. In \ref{Appendix analyticity propagator}, we prove that the (imaginary time) propagator is an analytic function of the coupling.


\section{\texorpdfstring{$\delta$}{delta}-Bose gas with repulsive interaction \texorpdfstring{($\kappa<0$)}{kappa<0}}
\label{Section repulsive}
\setcounter{equation}{0}
Let $S_{n}$ be the set of all $n!$ permutations of the integers between $1$ and $n$. In the following, we use the notations
\begin{equation}
\prod_{j<k}^{n}\equiv\prod_{1\leq j<k\leq n}\;.
\end{equation}
For $\kappa<0$ the TW formula states (in the notation of \cite{TW08.2} the strength of the potential is called $c=-\kappa$)
\begin{eqnarray}
\label{phi(kappa<0)}
&& \langle x|\rme^{-tH_{\kappa}}|y\rangle=\frac{1}{n!(2\pi)^{n}}\int_{\mathbb{R}^{n}}\rmd q_{1}\,\ldots\,\rmd q_{n}\,\nonumber\\
&& \hspace{20mm}\sum_{\sigma\in S_{n}}\!\!\!\!\!\prod_{\substack{j<k\\\sigma(j)>\sigma(k)}}^{n}\!\!\!\!\!\frac{q_{\sigma(j)}-q_{\sigma(k)}-\rmi\kappa}{q_{\sigma(j)}-q_{\sigma(k)}+\rmi\kappa}\;\prod_{j=1}^{n}\left(\rme^{\rmi q_{\sigma(j)}(x_{j}-y_{\sigma(j)})}\rme^{-tq_{j}^{2}}\right)\;.
\end{eqnarray}
The connection to the eigenfunction expansion can be seen by symmetrizing over all permutations of the $q_{j}$. We introduce a new permutation $\tau\in S_{n}$ and replace $q_{j}$ by $q_{\tau(j)}$. Replacing then $\sigma$ by $\tau^{-1}\circ\sigma$, one finds
\begin{eqnarray}
\label{phi(kappa<0) symmetrized}
&& \langle x|\rme^{-tH_{\kappa}}|y\rangle=\frac{1}{n!^{2}(2\pi)^{n}}\sum_{\sigma,\tau\in S_{n}}\int_{\mathbb{R}^{n}}\rmd q_{1}\,\ldots\,\rmd q_{n}\,\nonumber\\
&&\hspace{20mm} \prod_{\substack{j,k=1\\\tau^{-1}(j)<\tau^{-1}(k)\\\sigma^{-1}(j)>\sigma^{-1}(k)}}^{n}\!\!\!\!\!\frac{q_{j}-q_{k}+\rmi\kappa}{q_{j}-q_{k}-\rmi\kappa}\;\prod_{j=1}^{n}\left(\rme^{\rmi q_{j}(x_{\sigma^{-1}(j)}-y_{\tau^{-1}(j)})}\rme^{-tq_{j}^{2}}\right)\;.
\end{eqnarray}
One notes the factorization
\begin{eqnarray}
\label{factorization kappa<0}
&&\fl\hspace{10mm} \prod_{\substack{j,k=1\\\tau^{-1}(j)<\tau^{-1}(k)\\\sigma^{-1}(j)>\sigma^{-1}(k)}}^{n}\!\!\!\!\!\frac{q_{j}-q_{k}+\rmi\kappa}{q_{j}-q_{k}-\rmi\kappa}\;=\!\!\!\!\!\!\prod_{\substack{j<k\\\tau^{-1}(j)<\tau^{-1}(k)\\\sigma^{-1}(j)>\sigma^{-1}(k)}}^{n}\!\!\!\!\!\frac{q_{j}-q_{k}+\rmi\kappa}{q_{j}-q_{k}-\rmi\kappa}\!\!\!\!\!\prod_{\substack{j<k\\\tau^{-1}(j)>\tau^{-1}(k)\\\sigma^{-1}(j)<\sigma^{-1}(k)}}^{n}\!\!\!\!\!\frac{q_{j}-q_{k}-\rmi\kappa}{q_{j}-q_{k}+\rmi\kappa}\nonumber\\
&&\fl\hspace{54.5mm} =\!\!\!\!\!\!\prod_{\substack{j<k\\\sigma^{-1}(j)>\sigma^{-1}(k)}}^{n}\!\!\!\!\!\frac{q_{j}-q_{k}+\rmi\kappa}{q_{j}-q_{k}-\rmi\kappa}\!\!\!\!\!\prod_{\substack{j<k\\\tau^{-1}(j)>\tau^{-1}(k)}}^{n}\!\!\!\!\!\frac{q_{j}-q_{k}-\rmi\kappa}{q_{j}-q_{k}+\rmi\kappa}\;.
\end{eqnarray}
We introduce the Bethe eigenstates of the Hamiltonian (\ref{H kappa}) with repulsive interaction
\begin{equation}
\psi(x;q)=\frac{1}{n!}\sum_{\sigma\in S_{n}}\!\!\!\!\!\prod_{\substack{j<k\\\sigma^{-1}(j)>\sigma^{-1}(k)}}^{n}\!\!\!\!\!\frac{q_{j}-q_{k}+\rmi\kappa}{q_{j}-q_{k}-\rmi\kappa}\;\prod_{j=1}^{n}\rme^{\rmi q_{j}x_{\sigma^{-1}(j)}}\;,
\end{equation}
with momenta $q=(q_{1},\ldots,q_{n})\in\mathbb{R}^{n}$. The eigenstate $\psi(x;q)$ has energy
\begin{equation}
E(q)=\sum_{j=1}^{n}q_{j}^{2}\;.
\end{equation}
In terms of Bethe eigenstates, Eq. (\ref{phi(kappa<0) symmetrized}), (\ref{factorization kappa<0}) rewrite as
\begin{equation}
\label{phi(kappa<0) Bethe}
\langle x|\rme^{-tH_{\kappa}}|y\rangle=\frac{1}{(2\pi)^{n}}\int_{\mathbb{R}^{n}}\rmd q_{1}\,\ldots\,\rmd q_{n}\,\psi(x;q)\,\overline{\psi(y;q)}\;\rme^{-tE(q)}\;,
\end{equation}
where $\overline{(\ldots)}$ denotes complex conjugation. At $t=0$, (\ref{phi(kappa<0) Bethe}) reduces to the completeness relation for the Bethe eigenstates,
\begin{equation}
\openone=\frac{1}{(2\pi)^{n}}\int_{\mathbb{R}^{n}}\rmd q_{1}\,\ldots\,\rmd q_{n}\,|\psi(q)\rangle\,\langle\psi(q)|\;.
\end{equation}
A proof of the orthonormality of the Bethe eigenstates can be found \textit{e.g.} in \cite{D10.2}, Appendix A.


\section{\texorpdfstring{$\delta$}{delta}-Bose gas with attractive interaction \texorpdfstring{($\kappa>0$)}{kappa>0}}
\label{Section attractive}
\setcounter{equation}{0}
As shown in \ref{Appendix analyticity propagator}, the propagator is an analytic function of the coupling $\kappa$ for $t>0$. By analytic continuation of (\ref{phi(kappa<0)}) from $\kappa<0$ to $\kappa>0$, we will derive in Section \ref{Section attractive} an exact expression for the propagator in the attractive case $\kappa>0$.

Before stating the two main theorems, a few definitions are needed. We call $D_{n,M}$ the set of the $M$-tuples $\vec{n}=(n_{1},\ldots,n_{M})$ such that $n_{j}\geq1$, $j=1,\ldots,M$ and $n_{1}+\ldots+n_{M}=n$. For $\vec{n}\in D_{n,M}$, the \textit{clusters} $\Omega_{j}(\vec{n})\equiv\Omega_{j}$, $j=1,\ldots,M$ are defined by
\begin{equation}
\label{Omegaj}
\Omega_{j}=\{n_{1}+\ldots+n_{j-1}+1,\ldots,n_{1}+\ldots+n_{j}\}\;.
\end{equation}
From the Bethe ansatz point of view, the clusters will correspond to bound states of the particles. The function $r_{\vec{n}}\equiv r$, acting on $\{1,\ldots,n\}$, is defined by
\begin{equation}
\label{rn(a)}
r(a)=s \quad\text{for}\quad a=n_{1}+\ldots+n_{j-1}+s\in\Omega_{j}\;.
\end{equation}
More visually, one has
\begin{equation}
\fl\hspace{10mm} \begin{array}{r|ccc|ccc|c|ccc|}
a & 1 & \ldots & n_{1} & n_{1}+1 & \ldots & n_{1}+n_{2} & \ldots & n-n_{M}+1 & \ldots & n\\\hline
\text{Cluster} & & \Omega_{1} & & & \Omega_{2} & & \ldots & & \Omega_{M} & \\
r(a) & 1 & \ldots & n_{1} & 1 & \ldots & n_{2} & \ldots & 1 & \ldots & n_{M}
\end{array}\;.
\end{equation}
Finally, we call $S_{n}'(\vec{n})$ (respectively $S_{n}''(\vec{n})$) the subset of $S_{n}$ containing only the permutations $\sigma$ (resp. $\tau$) such that for all $j=1,\ldots,M$ and $a,b\in\Omega_{j}$ with $a<b$ one has $\sigma^{-1}(a)<\sigma^{-1}(b)$ (resp. $\tau^{-1}(a)>\tau^{-1}(b)$).

In the attractive case, the following expression for the propagator is proved in Section \ref{Section analytic continuation}.
\begin{theorem}
\label{Theorem propagator 1}
For fixed $\vec{n}$, let $\mu_{j}$, $j=1,\ldots,M$, be arbitrary real numbers satisfying the constraint
\begin{equation}
-n_{j}<\mu_{j}\leq0\;.
\end{equation}
For $\kappa>0$ and $x,y\in\Lambda$, one has
\begin{eqnarray}
\label{phi(kappa>0)}
&&\fl \langle x|\rme^{-tH_{\kappa}}|y\rangle=\sum_{M=1}^{n}\frac{\kappa^{n-M}}{n!M!(2\pi)^{M}}\sum_{\vec{n}\in D_{n,M}}\prod_{j=1}^{M}(n_{j}!(n_{j}-1)!)\int_{\mathbb{R}^{M}}\rmd q_{1}\,\ldots\,\rmd q_{M}\,\nonumber\\
&&\fl\hspace{20.5mm} \sum_{\sigma\in S_{n}'(\vec{n})}\sum_{\tau\in S_{n}''(\vec{n})}\prod_{j=1}^{M}\prod_{a\in\Omega_{j}}\left(\rme^{\rmi(q_{j}+\rmi\kappa(\mu_{j}+r(a)-1))(x_{\sigma^{-1}(a)}-y_{\tau^{-1}(a)})}\rme^{-t(q_{j}+\rmi\kappa(\mu_{j}+r(a)-1))^{2}}\right)\nonumber\\
&&\fl\hspace{20mm} \times\prod_{\substack{j,k=1\\j\neq k}}^{M}\!\!\!\!\!\!\!\!\prod_{\substack{a\in\Omega_{j}\\b\in\Omega_{k}\\\sigma^{-1}(a)>\sigma^{-1}(b)\\\tau^{-1}(a)<\tau^{-1}(b)}}\!\!\!\!\!\!\!\!\left(\frac{(q_{j}+\rmi\kappa(\mu_{j}+r(a)))-(q_{k}+\rmi\kappa(\mu_{k}+r(b)))+\rmi\kappa}{(q_{j}+\rmi\kappa(\mu_{j}+r(a)))-(q_{k}+\rmi\kappa(\mu_{k}+r(b)))-\rmi\kappa}\right)\;.
\end{eqnarray}
All the apparent poles in the integrand cancel except for simple poles at $q_{j}+\rmi\kappa\mu_{j}=q_{k}+\rmi\kappa(\mu_{k}+n_{k})$ and $q_{j}+\rmi\kappa(\mu_{j}+n_{j})=q_{k}+\rmi\kappa\mu_{k}$, $j<k$. The integrand vanishes at $q_{j}+\rmi\kappa\mu_{j}=q_{k}+\rmi\kappa\mu_{k}$ and $q_{j}+\rmi\kappa(\mu_{j}+n_{j})=q_{k}+\rmi\kappa(\mu_{k}+n_{k})$.
\end{theorem}
Compared to the TW formula (\ref{phi(kappa<0)}), our result (\ref{phi(kappa>0)}) is more complicated. Physically, the complication can be traced to the presence of bound states for the attractive case: the $n$ particles are arranged in $M$ clusters of size $n_{1}$, \ldots, $n_{M}$, hence the extra summations over $M$ and $\vec{n}$. Furthermore, Eq. (\ref{phi(kappa>0)}) contains a summation over two permutations $\sigma$ and $\tau$ instead of only one for the TW formula (\ref{phi(kappa<0)}). In the special case $x=y=0$ discussed in Section \ref{Section x=y=0}, both summations over $\sigma$ and $\tau$ can be eliminated.

For the attractive case, the propagator can also be written in terms of a summation over the eigenstates of the Hamiltonian (\ref{H kappa}). The Bethe eigenfunctions for attractive interaction are (see \cite{D10.2}, Eq. (B.26) and (B.48); in \cite{D10.2}, $r(a)$ is equal to $r(\sigma(a))$ with our notations, and $\Omega_{j}$ to $\sigma^{-1}(\Omega_{j})$)
\begin{eqnarray}
\label{psi(M,n,q)}
&&\fl\hspace{5mm} \psi(x;M,\vec{n},q)=\frac{\kappa^{\frac{n-M}{2}}}{\sqrt{n!}}\prod_{j=1}^{M}\sqrt{n_{j}!(n_{j}-1)!}\sum_{\sigma\in S_{n}'(\vec{n})}\prod_{j=1}^{M}\prod_{a\in\Omega_{j}}\left(\rme^{\rmi\left(q_{j}+\rmi\kappa\left(r(a)-\frac{n_{j}}{2}-\frac{1}{2}\right)\right)x_{\sigma^{-1}(a)}}\right)\nonumber\\
&&\fl\hspace{20mm} \times\prod_{j<k}^{M}\!\!\!\!\!\!\!\!\prod_{\substack{a\in\Omega_{j}\\b\in\Omega_{k}\\\sigma^{-1}(a)>\sigma^{-1}(b)}}\!\!\!\!\!\!\!\!\left(\frac{(q_{j}+\rmi\kappa(r(a)-\frac{n_{j}}{2}))-(q_{k}+\rmi\kappa(r(b)-\frac{n_{k}}{2}))+\rmi\kappa}{(q_{j}+\rmi\kappa(r(a)-\frac{n_{j}}{2}))-(q_{k}+\rmi\kappa(r(b)-\frac{n_{k}}{2}))-\rmi\kappa}\right)\;,
\end{eqnarray}
with $M=1,\ldots,n$, $\vec{n}\in D_{n,M}$ and $q\in\mathbb{R}^{M}$. Eq. (\ref{psi(M,n,q)}) is an eigenfunction of the Hamiltonian (\ref{H kappa}) with eigenvalue
\begin{equation}
\label{E(M,n,q)}
E(M,\vec{n},q)=\sum_{j=1}^{M}\left(n_{j}q_{j}^{2}-\frac{\kappa^{2}}{12}(n_{j}^{3}-n_{j})\right)\;.
\end{equation}
The relation of the propagator with the Bethe eigenfunctions is stated as next theorem in terms of (\ref{psi(M,n,q)}) and (\ref{E(M,n,q)}).
\begin{theorem}
\label{Theorem propagator 2}
For $\kappa>0$ and $x,y\in\Lambda$, one has
\begin{eqnarray}
\label{phi(kappa>0,Bethe)}
&& \langle x|\rme^{-tH_{\kappa}}|y\rangle=\sum_{M=1}^{n}\frac{1}{M!(2\pi)^{M}}\sum_{\vec{n}\in D_{n,M}}\int_{\mathbb{R}^{M}}\rmd q_{1}\,\ldots\,\rmd q_{M}\,\nonumber\\
&&\hspace{25mm} \psi(x;M,\vec{n},q)\;\overline{\psi(y;M,\vec{n},q)}\;\rme^{-tE(M,\vec{n},q)}\;.
\end{eqnarray}
\end{theorem}
As in the case of repulsive interaction discussed in Section \ref{Section repulsive}, taking $t=0$ yields the completeness relation for the Bethe eigenstates (\ref{psi(M,n,q)}). Their orthonormality is proved in \cite{D10.2}, Appendix B.


\section{Analytic continuation from \texorpdfstring{$\kappa<0$}{kappa<0} to \texorpdfstring{$\kappa>0$}{kappa>0}}
\label{Section analytic continuation}
\setcounter{equation}{0}
In this section, the TW formula (\ref{phi(kappa<0)}) for the propagator is extended by analytic continuation to the attractive case $\kappa>0$. Theorem \ref{Theorem propagator 1} and Theorem \ref{Theorem propagator 2} are proved.

\subsection{Contribution of the residues}
The contours of integration in (\ref{phi(kappa<0)}) can be moved freely as long as the denominators $q_{j}-q_{k}-\rmi\kappa$ keep a strictly positive imaginary part. In particular, if the integration is shifted to $q_{j}\in\mathbb{R}+\rmi\lambda(n-j)$, $j=1,\ldots,n$, with $\lambda>0$, we obtain a formula valid for all $\kappa$ such that $\Im(q_{j}-q_{k}-\rmi\kappa)=(k-j)\lambda-\kappa>0$ for $j<k$, \textit{i.e.} for all $\kappa<\lambda$. One obtains
\begin{eqnarray}
\label{phi(l=n-1)}
&& \langle x|\rme^{-tH_{\kappa}}|y\rangle=\frac{1}{n!(2\pi)^{n}}\prod_{a=1}^{n}\left(\int_{\mathbb{R}+\rmi\lambda(n-a)}\!\!\!\!\!\rmd q_{a}\right)\nonumber\\
&&\hspace{20mm} \sum_{\sigma\in S_{n}}\!\!\!\!\!\prod_{\substack{a<b\\\sigma^{-1}(a)>\sigma^{-1}(b)}}^{n}\!\!\!\!\!\frac{q_{a}-q_{b}+\rmi\kappa}{q_{a}-q_{b}-\rmi\kappa}\;\prod_{a=1}^{n}\left(\rme^{\rmi q_{a}(x_{\sigma^{-1}(a)}-y_{a})}\rme^{-tq_{a}^{2}}\right)\;.
\end{eqnarray}
In the following, we want to further move the contours of integration, but this time the contours will have to cross poles of the integrand, which will add several new terms resulting from the residues at these poles, symbolically,
\begin{equation}
\begin{array}{c}
\begin{picture}(100,23)
\put(0,20){\line(1,0){40}}\put(0,20){\vector(1,0){20}}\put(20,10){$\times$}
\put(50,10){$=$}
\put(60,0){\line(1,0){40}}\put(60,0){\vector(1,0){20}}\put(80,10){$\times$}\put(81.5,11){\circle{10}}\put(80.4,6.1){\vector(-1,0){0}}
\end{picture}
\end{array}
\end{equation}
If one denotes by $j\to k$ the action of taking the residue at $q_{j}=q_{k}+\rmi\kappa r$ with $r\in\mathbb{Z}$, the terms, obtained after moving the contours of integration, correspond to collections of $j\to k$ such that, for each $\ell=1,\ldots,n$, $\ell\to\ldots$ appears only once in the collection (since after taking the residue at $q_{\ell}=q_{m}+\rmi\kappa r$, the integrand no longer contains $q_{\ell}$). Each term thus corresponds to a forest (a set of trees), for example
\begin{equation}
\fl\hspace{20mm}
\begin{array}{ccc}
\{1\to2,2\to7,3\to5,4\to7\}
&
\qquad\Leftrightarrow\qquad
&
\begin{array}{c}
\begin{picture}(30,24)
\put(0,20){$1$}\put(0,10){$2$}\put(5,0){$7$}\put(10,10){$4$}\put(20,10){$3$}\put(20,0){$5$}\put(30,0){$6$}
\put(1,19){\vector(0,-1){5}}\put(2,9){\vector(1,-2){2.5}}\put(10,9){\vector(-1,-2){2.5}}
\put(21,9){\vector(0,-1){5}}
\end{picture}
\end{array}
\end{array}
\end{equation}
The particular trees obtained in this fashion depend on the order in which the contours are moved. Here, we choose to move first the contour for $q_{n-1}$ in such a way that it crosses only the pole at $q_{n-1}=q_{n}+\rmi\kappa$. Then, we move the contour for $q_{n-2}$ in such a way that it crosses only the poles at $q_{n-2}=q_{n}+\rmi\kappa$ and $q_{n-2}=q_{n-1}+\rmi\kappa$ (in which case we still have an integration over both $q_{n-1}$ and $q_{n}$), or only the pole at $q_{n-2}=q_{n}+2\rmi\kappa$ (in which case the residue at $q_{n-1}=q_{n}+\rmi\kappa$ has been taken). We continue in this fashion until in the final step the contour for $q_{1}$ is moved.

In principle, after moving the contours, the propagator $\langle x|\rme^{-tH_{\kappa}}|y\rangle$ will be expressed as a sum over forests. In fact, it turns out that during this procedure there are many cancellations which remove all the forests which contain trees with ``branches'': in other words, only the forests with merely ``branchless'' trees (like $a$, $a\to b$, $a\to b\to c$, $a\to b\to c\to d$, \ldots) remain after these cancellations. Instead of a sum over forests, we end up with a sum over partitions of $\{1,\ldots,n\}$ (each element of the partition corresponding to one of the branchless trees of the forest).

In the context of the distribution of the leftmost particle in the asymmetric simple exclusion process, the procedure described here bears some similarity with the transformation from Theorem 3.1 to Theorem 3.2 in \cite{TW08.1}, where contours of integration are moved from small to large circles. A complication in our context is that we need a summation over all partitions of $\{1,\ldots,n\}$ and not just over subsets of $\{1,\ldots,n\}$. We expect that in the case of the full transition probability for the asymmetric exclusion process an expression with an integration over large circles would require summing over all partitions of $\{1,\ldots,n\}$.

A proof of the previous statements is based on induction w.r.t. an integer $\ell$ such that all the contours for $q_{\ell+1}$, \ldots, $q_{n-1}$ have already been moved.

We introduce a few notations. For a boolean condition $c$, $\openone_{\{c\}}$ is defined to be equal to $1$ if $c$ is true and $0$ otherwise. For $\vec{n}\in D_{n,M}$, the set $P_{n}(\vec{n})$ contains all the partitions $\vec{A}=\{A_{1},\ldots,A_{M}\}$ of $\{1,\ldots,n\}$ with $|A_{j}|=n_{j}$, $j=1,\ldots,M$. The partition $\vec{A}$ verifies $A_{1}\cup\ldots\cup A_{M}=\{1,\ldots,n\}$ and for $j\neq k$ $A_{j}\cap A_{k}=\emptyset$. The partitions are not ordered, \textit{i.e.} the partition $\vec{B}=\{A_{R(1)},\ldots,A_{R(M)}\}$ is considered to be the same element of $P_{n}(\vec{n})$ as $\vec{A}$ for all $R\in S_{M}$. Each $A_{j}$ is called a \textit{cluster}, and will correspond to a bound state of particles in the Bethe ansatz point of view. For a partition $\vec{A}$, we define $d_{\vec{A}}(a)$, $a=1,\ldots,n$, (abbreviated as $d(a)$ to lighten the notation) to be the rank of $a$ in its cluster $A_{j}$, starting with rank $0$ for the largest element of the cluster, rank $1$ for the second largest, \ldots, and rank $|A_{j}|-1$ for the smallest element of $A_{j}$.

With these notation, the following lemma can be stated.
\begin{lemma}
\label{Lemma analytic continuation}
Let $\ell$ be an integer between $0$ and $n-1$. For fixed $M=1,\ldots,n$, let $\epsilon_{j}$, $j=\ell+1,\ldots,M$ be distinct numbers with $0\leq\epsilon_{j}<1$. Then, for $0<\kappa<\lambda$ one has
\begin{eqnarray}
\label{phi(l)}
&&\fl\hspace{10mm} \langle x|\rme^{-tH_{\kappa}}|y\rangle=\sum_{M=1}^{n}\frac{\kappa^{n-M}}{n!(2\pi)^{M}}\sum_{\vec{n}\in D_{n,M}}\prod_{j=1}^{M}(n_{j}!(n_{j}-1)!)\prod_{j=1}^{\ell}\left(\int_{\mathbb{R}+\rmi\lambda(n-j)}\!\!\!\!\!\rmd q_{j}\right)\nonumber\\
&&\fl\hspace{10mm} \times\!\!\prod_{j=\ell+1}^{M}\left(\int_{\mathbb{R}-\rmi\kappa\epsilon_{j}}\rmd q_{j}\right)\sum_{\vec{A}\in P_{n}(\vec{n})}\sum_{\sigma\in S_{n}}\;\prod_{j=1}^{\ell}\openone_{\{A_{j}=\{j\}\}}\;\prod_{j=1}^{M}\!\!\prod_{\substack{a,b\in A_{j}\\a<b}}\openone_{\{\sigma^{-1}(a)>\sigma^{-1}(b)\}}\nonumber\\
&&\fl\hspace{10mm} \times\prod_{j=1}^{M}\prod_{a\in A_{j}}\left(\rme^{\rmi(q_{j}+\rmi\kappa d(a))(x_{\sigma^{-1}(a)}-y_{a})}\rme^{-t(q_{j}+\rmi\kappa d(a))^{2}}\right)\nonumber\\
&&\fl\hspace{10mm} \times\!\!\prod_{\substack{j,k=1\\j\neq k}}^{M}\!\!\!\!\!\!\!\!\!\!\prod_{\substack{a\in A_{j}\\b\in A_{k}\\a<b\\\sigma^{-1}(a)>\sigma^{-1}(b)}}\!\!\!\!\!\!\!\!\left(\frac{(q_{j}+\rmi\kappa d(a))-(q_{k}+\rmi\kappa d(b))+\rmi\kappa}{(q_{j}+\rmi\kappa d(a))-(q_{k}+\rmi\kappa d(b))-\rmi\kappa}\right)\;.
\end{eqnarray}
\end{lemma}
\begin{proof}
The constraint on the $\epsilon_{j}$, $j=\ell+1,M$, implies that (\ref{phi(l)}) is well defined since all the poles are at $q_{j}=q_{k}+\rmi\delta$ with $\Im(q_{j})\neq\Im(q_{k})+\delta$.

For $\ell=n-1$, the identity between the expressions (\ref{phi(l=n-1)}) and (\ref{phi(l)}) of $\langle x|\rme^{-tH_{\kappa}}|y\rangle$ is immediate. All the clusters must have size $1$, and only $M=n$ contributes. Since the poles for $q_{n}$ are at $q_{n}=q_{j}-\rmi\kappa$, $j=1,\ldots,n-1$, the contour for $q_{n}$ can be moved freely from $\mathbb{R}$ to $\mathbb{R}-\rmi\kappa\epsilon_{n}$, provided $0\leq\epsilon_{n}<1$ and $\kappa>0$.

We now proceed to prove the general identity by induction in $\ell$: we assume that the expression (\ref{phi(l)}) for $\langle x|\rme^{-tH_{\kappa}}|y\rangle$ is valid for $\ell\geq1$ and will establish the expression with $\ell$ replaced by $\ell-1$.

For given $\vec{n}$, we want to move the contour of integration for $q_{\ell}$ from $\mathbb{R}+\rmi\lambda(n-\ell)$ to $\mathbb{R}-\rmi\kappa\epsilon_{\ell}$ with $0\leq\epsilon_{\ell}<1$ and $\epsilon_{\ell}$ different from all the other $\epsilon_{k}$, $k=\ell+1,\ldots,M$. In order to accomplish this, one needs to take into account the residues of the poles at $q_{\ell}=z$ for $-\kappa\epsilon_{\ell}<\Im(z)<\lambda(n-\ell)$. The only poles for $q_{\ell}$ are at $z=q_{j}-\rmi\kappa$, $j=1,\ldots,\ell-1$, and at $z=q_{m}+\rmi\kappa(d(c)+1)$, $m=\ell+1,\ldots,M$, $c\in A_{m}$. In the first case, using $\kappa<\lambda$ and $j\leq\ell-1$, one finds $\Im(z)>\lambda(n-\ell)$, which implies that these poles do not contribute when moving the contour for $q_{\ell}$. In the second case, using $0<\kappa<\lambda$, $0\leq\epsilon_{m}<1$ and $\ell+d(c)+1\leq n$, one has $-\kappa\epsilon_{\ell}\leq0<\Im(z)<\lambda(n-\ell)$, which implies that all these poles contribute a residue (with a factor $-2\rmi\pi$ corresponding to a clockwise contour integration).

Moving the contour for $q_{\ell}$ produces several terms: one term corresponding to the integration over $q_{\ell}\in\mathbb{R}-\rmi\kappa\epsilon_{\ell}$, for which the integrand still depends on $q_{\ell}$, and one term for each $c\in A_{m}$, $m=\ell+1,\ldots,M$, for which the residue at $q_{\ell}=q_{m}+\rmi\kappa(d(c)+1)$ has been taken. The latter term corresponds to merging the cluster $A_{\ell}=\{\ell\}$ and the cluster $A_{m}$. Assuming $\sigma^{-1}(\ell)>\sigma^{-1}(c)$ (otherwise, the pole vanishes), this term is equal to
\begin{eqnarray}
\label{residue}
&& (-2\rmi\pi)\frac{\kappa^{n-M}}{n!(2\pi)^{M}}\prod_{j=1}^{M}(n_{j}!(n_{j}-1)!)\prod_{j=1}^{\ell-1}\left(\int_{\mathbb{R}+\rmi\lambda(n-j)}\!\!\!\!\!\rmd q_{j}\right)\prod_{j=\ell+1}^{M}\left(\int_{\mathbb{R}-\rmi\kappa\epsilon_{j}}\rmd q_{j}\right)\nonumber\\
&& \sum_{\vec{A}\in P_{n}(\vec{n})}\sum_{\sigma\in S_{n}}\prod_{j=1}^{\ell}\openone_{\{A_{j}=\{j\}\}}\prod_{j=1}^{M}\prod_{\substack{a,b\in A_{j}\\a<b}}\openone_{\{\sigma^{-1}(a)>\sigma^{-1}(b)\}}\nonumber\\
&& \times\left(\rme^{\rmi(q_{m}+\rmi\kappa(d(c)+1))(x_{\sigma^{-1}(\ell)}-y_{\ell})}\;\rme^{-t(q_{m}+\rmi\kappa(d(c)+1))^{2}}\right)\nonumber\\
&& \times\prod_{\substack{j=1\\j\neq\ell}}^{M}\prod_{a\in A_{j}}\left(\rme^{\rmi(q_{j}+\rmi\kappa d(a))(x_{\sigma^{-1}(a)}-y_{a})}\rme^{-t(q_{j}+\rmi\kappa d(a))^{2}}\right)\nonumber\\
&& \times\prod_{\substack{j,k=1\\j\neq k\\j,k\neq\ell}}^{M}\!\!\!\!\!\!\!\!\prod_{\substack{a\in A_{j}\\b\in A_{k}\\a<b\\\sigma^{-1}(a)>\sigma^{-1}(b)}}\!\!\!\!\!\!\!\!\left(\frac{(q_{j}+\rmi\kappa d(a))-(q_{k}+\rmi\kappa d(b))+\rmi\kappa}{(q_{j}+\rmi\kappa d(a))-(q_{k}+\rmi\kappa d(b))-\rmi\kappa}\right)\nonumber\\
&& \times\prod_{\substack{j,k=1\\j\neq k\\j\neq\ell}}^{M}\!\!\!\!\!\!\prod_{\substack{a\in A_{j}\\a<\ell\\\sigma^{-1}(a)>\sigma^{-1}(\ell)}}\!\!\!\!\!\!\left(\frac{(q_{j}+\rmi\kappa d(a))-(q_{m}+\rmi\kappa(d(c)+1))+\rmi\kappa}{(q_{j}+\rmi\kappa d(a))-(q_{m}+\rmi\kappa(d(c)+1))-\rmi\kappa}\right)\nonumber\\
&& \times\prod_{\substack{j,k=1\\j\neq k\\k\neq\ell}}^{M}\!\!\!\!\!\!\prod_{\substack{b\in A_{k}\\\ell<b\\\sigma^{-1}(\ell)>\sigma^{-1}(b)}}\!\!\!\!\!\!\left(\frac{(q_{m}+\rmi\kappa(d(c)+1))-(q_{k}+\rmi\kappa d(b))+\rmi\kappa}{(q_{m}+\rmi\kappa(d(c)+1))-(q_{k}+\rmi\kappa d(b))-\rmi\kappa}\right)\nonumber\\
&& \times\!\!\!\!\!\!\!\!\prod_{\substack{b\in A_{m}\\b\neq c\\\sigma^{-1}(\ell)>\sigma^{-1}(b)}}\!\!\!\!\!\!\!\!\left(\frac{(q_{m}+\rmi\kappa(d(c)+1))-(q_{m}+\rmi\kappa d(b))+\rmi\kappa}{(q_{m}+\rmi\kappa(d(c)+1))-(q_{m}+\rmi\kappa d(b))-\rmi\kappa}\right)\nonumber\\
&&\hspace{11mm} \times((q_{m}+\rmi\kappa(d(c)+1))-(q_{m}+\rmi\kappa d(c))+\rmi\kappa)\;.
\end{eqnarray}
The last line of (\ref{residue}) contributes a factor $2\rmi\kappa$ and the line before contributes
\begin{equation}
\label{residue factor 1}
\prod_{\substack{b\in A_{m}\\b\neq c\\\sigma^{-1}(\ell)>\sigma^{-1}(b)}}\!\!\!\!\!\!\!\!\left(\frac{d(c)-d(b)+2}{d(c)-d(b)}\right)\;.
\end{equation}
Let us first assume that $c=\min(A_{m})$. Then, for all $b\in A_{m}$, $b\neq c$ one has $\sigma^{-1}(b)<\sigma^{-1}(c)$. Together with $\sigma^{-1}(c)<\sigma^{-1}(\ell)$, it implies that all the elements of the cluster $A_{m}$ (except $c$) contribute in (\ref{residue factor 1}). This results in a factor $(n_{m}+1)n_{m}/2$. Combined with $(-2\rmi\pi)$ and $2\rmi\kappa$, we obtain a factor $2\pi\kappa(n_{k}+1)n_{k}$. The term with $c=\min(A_{m})$ thus corresponds exactly to the term of (\ref{phi(l)}) with $\ell$ replaced by $\ell-1$ and the partition $\vec{A}$ replaced by $\vec{B}$, obtained from $\vec{A}$ by merging the cluster $\{\ell\}$ with $A_{m}$ (after a renaming of the $q_{j}$, $n_{j}$, $\epsilon_{j}$ to $q_{j-1}$, $n_{j-1}$, $\epsilon_{j-1}$ for $\ell+1\leq j\leq M$).

It remains to show that for $c\neq\min(A_{m})$, the residues cancel each other. Since the $\sigma^{-1}(b)$ are ordered in the same way as the $d(b)$ for $b\in A_{m}$, there exists a unique number $f\in A_{m}$ such that for $b\in A_{m}$, if $b\geq f$ then $\sigma^{-1}(b)<\sigma^{-1}(\ell)$, and if $b<f$ then $\sigma^{-1}(b)>\sigma^{-1}(\ell)$. Since $\sigma^{-1}(c)<\sigma^{-1}(\ell)$, one has necessarily $f\leq c$ (or equivalently $d(f)\geq d(c)$). Then, (\ref{residue factor 1}) rewrites
\begin{equation}
\label{residue factor 2}
\prod_{\substack{b\in A_{m}\\b\neq c\\b\geq f}}\left(\frac{d(c)-d(b)+2}{d(c)-d(b)}\right)\;.
\end{equation}
The rest of the argument depends on the relative values of $d(c)$ and $d(f)$. If $d(f)\geq d(c)+2$, then, there exists $b\in A_{m}$ such that $b\geq f$ and $d(b)=d(c)+2$, thus (\ref{residue factor 2}) is equal to zero. Since $d(f)\geq d(c)$, the only cases left are $d(f)=d(c)+1$ and $d(f)=d(c)$, for which (\ref{residue factor 2}) rewrites respectively
\begin{equation}
\label{residue factor 3}
\fl\hspace{5mm} \frac{d(c)-d(f)+2}{d(c)-d(f)}\prod_{\substack{b\in A_{m}\\b>c}}\left(\frac{d(c)-d(b)+2}{d(c)-d(b)}\right)=-\frac{(d(c)+1)(d(c)+2)}{2}\;,
\end{equation}
and
\begin{equation}
\label{residue factor 4}
\prod_{\substack{b\in A_{m}\\b>c}}\left(\frac{d(c)-d(b)+2}{d(c)-d(b)}\right)=\frac{(d(c)+1)(d(c)+2)}{2}\;.
\end{equation}
Let us call $c'$ the element of $A_{m}$ such that $d(c')=d(c)+1$ ($c'$ is the smallest element of $A_{m}$ larger that $c$). One notes that the two previous cases are exchanged when replacing $\sigma^{-1}$ by $\sigma^{-1}\circ\theta_{\ell,c'}$, with $\theta_{\ell,c'}$ the permutation exchanging $\ell$ and $c'$. Thus, summing over all permutations $\sigma$, the residues at $q_{\ell}=q_{m}+\rmi\kappa(d(c)+1)$ cancel.
\end{proof}
Our construction achieves the proof of (\ref{phi(l)}) for $0<\kappa<\lambda$ and $\ell$ between $0$ and $n-1$. In particular, for $\ell=0$, one has
\begin{eqnarray}
\label{phi(l=0)}
&& \langle x|\rme^{-tH_{\kappa}}|y\rangle=\sum_{M=1}^{n}\frac{\kappa^{n-M}}{n!(2\pi)^{M}}\sum_{\vec{n}\in D_{n,M}}\prod_{j=1}^{M}(n_{j}!(n_{j}-1)!)\nonumber\\
&& \times\prod_{j=1}^{M}\left(\int_{\mathbb{R}-\rmi\kappa\epsilon_{j}}\rmd q_{j}\right)\sum_{\vec{A}\in P_{n}(\vec{n})}\sum_{\sigma\in S_{n}}\prod_{j=1}^{M}\prod_{\substack{a,b\in A_{j}\\a<b}}\openone_{\{\sigma^{-1}(a)>\sigma^{-1}(b)\}}\nonumber\\
&& \times\prod_{j=1}^{M}\prod_{a\in A_{j}}\left(\rme^{\rmi(q_{j}+\rmi\kappa d(a))(x_{\sigma^{-1}(a)}-y_{a})}\rme^{-t(q_{j}+\rmi\kappa d(a))^{2}}\right)\nonumber\\
&& \times\!\!\!\!\prod_{\substack{j,k=1\\j\neq k}}^{M}\!\!\!\!\!\!\!\!\prod_{\substack{a\in A_{j}\\b\in A_{k}\\a<b\\\sigma^{-1}(a)>\sigma^{-1}(b)}}\!\!\!\!\!\!\!\!\left(\frac{(q_{j}+\rmi\kappa d(a))-(q_{k}+\rmi\kappa d(b))+\rmi\kappa}{(q_{j}+\rmi\kappa d(a))-(q_{k}+\rmi\kappa d(b))-\rmi\kappa}\right)\;.
\end{eqnarray}
This expression no longer depends on $\lambda$. Hence, it is valid in the entire range $\kappa>0$.

\subsection{Partitions and permutations}
Exchanging $A_{j}$ and $A_{k}$ in (\ref{phi(l=0)}) is the same as exchanging $q_{j}$ and $q_{k}$, or $\epsilon_{j}$ and $\epsilon_{k}$. Since the $\epsilon_{j}$ are arbitrary numbers satisfying a constraint ($0\leq\epsilon_{j}<1$ and all $\epsilon_{j}$ different) which is the same for all $j$, it is possible to add an extra sum over all permutations of the $A_{j}$, compensated by a global factor $1/M!$. This is equivalent to summing now over \textit{ordered partitions} $\vec{A}=(A_{1},\ldots,A_{n})$, such that for all $R\in S_{M}$ different from the identity permutation, the ordered partition $\vec{B}=(A_{R(1)},\ldots,A_{R(M)})$ is distinct from $\vec{A}$.

There exists a bijection between ordered partitions $\vec{A}$ such that $|A_{j}|=n_{j}$, $j=1,\ldots,M$, and permutations $\tau\in S_{n}''(\vec{n})$ ($S_{n}''(\vec{n})$ is defined after Eq. (\ref{phi(kappa>0)})). By this bijection, the cluster $A_{j}$ is equal to $\{\tau^{-1}(a),a\in\Omega_{j}(\vec{n})\}\equiv\tau^{-1}(\Omega_{j})$ and one has $1+d_{\vec{A}}(a)=r_{\vec{n}}(\tau(a))$, using the definitions (\ref{Omegaj}) and (\ref{rn(a)}). Eq. (\ref{phi(l=0)}) becomes
\begin{eqnarray}
&& \langle x|\rme^{-tH_{\kappa}}|y\rangle=\sum_{M=1}^{n}\frac{\kappa^{n-M}}{n!M!(2\pi)^{M}}\sum_{\vec{n}\in D_{n,M}}\prod_{j=1}^{M}(n_{j}!(n_{j}-1)!)\nonumber\\
&& \times\prod_{j=1}^{M}\left(\int_{\mathbb{R}-\rmi\kappa\epsilon_{j}}\rmd q_{j}\right)\sum_{\sigma\in S_{n}}\sum_{\tau\in S_{n}''(\vec{n})}\prod_{j=1}^{M}\prod_{\substack{a,b\in\tau^{-1}(\Omega_{j})\\a<b}}\openone_{\{\sigma^{-1}(a)>\sigma^{-1}(b)\}}\nonumber\\
&& \times\prod_{j=1}^{M}\prod_{a\in\tau^{-1}(\Omega_{j})}\left(\rme^{\rmi(q_{j}+\rmi\kappa(r(\tau(a))-1))(x_{\sigma^{-1}(a)}-y_{a})}\rme^{-t(q_{j}+\rmi\kappa(r(\tau(a))-1))^{2}}\right)\nonumber\\
&& \times\!\!\!\!\prod_{\substack{j,k=1\\j\neq k}}^{M}\!\!\!\!\!\!\!\!\prod_{\substack{a\in\tau^{-1}(\Omega_{j})\\b\in\tau^{-1}(\Omega_{k})\\a<b\\\sigma^{-1}(a)>\sigma^{-1}(b)}}\!\!\!\!\!\!\!\!\left(\frac{(q_{j}+\rmi\kappa r(\tau(a)))-(q_{k}+\rmi\kappa r(\tau(b)))+\rmi\kappa}{(q_{j}+\rmi\kappa r(\tau(a)))-(q_{k}+\rmi\kappa r(\tau(b)))-\rmi\kappa}\right)\;.
\end{eqnarray}
One can now replace everywhere $a$ and $b$ by $\tau^{-1}(a)$ and $\tau^{-1}(b)$. We also replace $\sigma$ by $\tau^{-1}\circ\sigma$. Because of the definition of $S_{n}''(\vec{n})$, if $a,b\in\Omega_{j}$ with $\tau^{-1}(a)<\tau^{-1}(b)$ then $a>b$. This implies that the constraint with the $\openone_{\{\ldots\}}$ is equivalent to $\sigma\in S_{n}'(\vec{n})$ (defined after Eq. (\ref{phi(kappa>0)})). We obtain
\begin{eqnarray}
\label{phi(sigma,tau)}
&&\fl \langle x|\rme^{-tH_{\kappa}}|y\rangle=\sum_{M=1}^{n}\frac{\kappa^{n-M}}{n!M!(2\pi)^{M}}\sum_{\vec{n}\in D_{n,M}}\prod_{j=1}^{M}(n_{j}!(n_{j}-1)!)\prod_{j=1}^{M}\left(\int_{\mathbb{R}-\rmi\kappa\epsilon_{j}}\rmd q_{j}\right)\nonumber\\
&&\fl\hspace{25mm} \sum_{\sigma\in S_{n}'(\vec{n})}\sum_{\tau\in S_{n}''(\vec{n})}\prod_{j=1}^{M}\prod_{a\in\Omega_{j}}\left(\rme^{\rmi(q_{j}+\rmi\kappa(r(a)-1))(x_{\sigma^{-1}(a)}-y_{\tau^{-1}(a)})}\rme^{-t(q_{j}+\rmi\kappa(r(a)-1))^{2}}\right)\nonumber\\
&&\fl\hspace{25mm} \times\prod_{\substack{j,k=1\\j\neq k}}^{M}\!\!\!\!\!\!\!\!\prod_{\substack{a\in\Omega_{j}\\b\in\Omega_{k}\\\sigma^{-1}(a)>\sigma^{-1}(b)\\\tau^{-1}(a)<\tau^{-1}(b)}}\!\!\!\!\!\!\!\!\left(\frac{(q_{j}+\rmi\kappa r(a))-(q_{k}+\rmi\kappa r(b))+\rmi\kappa}{(q_{j}+\rmi\kappa r(a))-(q_{k}+\rmi\kappa r(b))-\rmi\kappa}\right)\;.
\end{eqnarray}

\subsection{Pole structure of the integrand and summation over Bethe eigenstates}
One has the factorization
\begin{eqnarray}
\label{factorization kappa>0}
&& \prod_{\substack{j,k=1\\j\neq k}}^{M}\!\!\!\!\!\!\!\!\prod_{\substack{a\in\Omega_{j}\\b\in\Omega_{k}\\\sigma^{-1}(a)>\sigma^{-1}(b)\\\tau^{-1}(a)<\tau^{-1}(b)}}\!\!\!\!\!\!\!\!\left(\frac{(q_{j}+\rmi\kappa r(a))-(q_{k}+\rmi\kappa r(b))+\rmi\kappa}{(q_{j}+\rmi\kappa r(a))-(q_{k}+\rmi\kappa r(b))-\rmi\kappa}\right)\nonumber\\
&& =\prod_{j<k}^{M}\!\!\!\!\!\!\!\!\prod_{\substack{a\in\Omega_{j}\\b\in\Omega_{k}\\\sigma^{-1}(a)>\sigma^{-1}(b)\\\tau^{-1}(a)<\tau^{-1}(b)}}\!\!\!\!\!\!\!\!\left(\frac{(q_{j}+\rmi\kappa r(a))-(q_{k}+\rmi\kappa r(b))+\rmi\kappa}{(q_{j}+\rmi\kappa r(a))-(q_{k}+\rmi\kappa r(b))-\rmi\kappa}\right)\nonumber\\
&&\hspace{20mm} \times\prod_{j<k}^{M}\!\!\!\!\!\!\!\!\prod_{\substack{a\in\Omega_{j}\\b\in\Omega_{k}\\\sigma^{-1}(a)<\sigma^{-1}(b)\\\tau^{-1}(a)>\tau^{-1}(b)}}\!\!\!\!\!\!\!\!\left(\frac{(q_{j}+\rmi\kappa r(a))-(q_{k}+\rmi\kappa r(b))-\rmi\kappa}{(q_{j}+\rmi\kappa r(a))-(q_{k}+\rmi\kappa r(b))+\rmi\kappa}\right)\nonumber\\
&& =\prod_{j<k}^{M}\!\!\!\!\!\!\!\!\prod_{\substack{a\in\Omega_{j}\\b\in\Omega_{k}\\\sigma^{-1}(a)>\sigma^{-1}(b)}}\!\!\!\!\!\!\!\!\left(\frac{(q_{j}+\rmi\kappa r(a))-(q_{k}+\rmi\kappa r(b))+\rmi\kappa}{(q_{j}+\rmi\kappa r(a))-(q_{k}+\rmi\kappa r(b))-\rmi\kappa}\right)\nonumber\\
&&\hspace{20mm} \times\prod_{j<k}^{M}\!\!\!\!\!\!\!\!\prod_{\substack{a\in\Omega_{j}\\b\in\Omega_{k}\\\tau^{-1}(a)>\tau^{-1}(b)}}\!\!\!\!\!\!\!\!\left(\frac{(q_{j}+\rmi\kappa r(a))-(q_{k}+\rmi\kappa r(b))-\rmi\kappa}{(q_{j}+\rmi\kappa r(a))-(q_{k}+\rmi\kappa r(b))+\rmi\kappa}\right)\;.
\end{eqnarray}
We introduce the functions
\begin{eqnarray}
\label{phi(x)}
&&\fl\hspace{10mm} \varphi_{\kappa}(x;M,\vec{n},q)=\frac{\kappa^{\frac{n-M}{2}}}{\sqrt{n!}}\prod_{j=1}^{M}\sqrt{n_{j}!(n_{j}-1)!}\sum_{\sigma\in S_{n}'(\vec{n})}\prod_{j=1}^{M}\prod_{a\in\Omega_{j}}\left(\rme^{\rmi(q_{j}+\rmi\kappa(r(a)-1))x_{\sigma^{-1}(a)}}\right)\nonumber\\
&&\fl\hspace{35mm} \times\prod_{j<k}^{M}\!\!\!\!\!\!\!\!\prod_{\substack{a\in\Omega_{j}\\b\in\Omega_{k}\\\sigma^{-1}(a)>\sigma^{-1}(b)}}\!\!\!\!\!\!\!\!\left(\frac{(q_{j}+\rmi\kappa r(a))-(q_{k}+\rmi\kappa r(b))+\rmi\kappa}{(q_{j}+\rmi\kappa r(a))-(q_{k}+\rmi\kappa r(b))-\rmi\kappa}\right)\;,
\end{eqnarray}
and
\begin{eqnarray}
&&\fl\hspace{10mm} \widetilde{\varphi}_{\kappa}(y;M,\vec{n},q)=\frac{\kappa^{\frac{n-M}{2}}}{\sqrt{n!}}\prod_{j=1}^{M}\sqrt{n_{j}!(n_{j}-1)!}\sum_{\tau\in S_{n}''(\vec{n})}\prod_{j=1}^{M}\prod_{a\in\Omega_{j}}\left(\rme^{-\rmi(q_{j}+\rmi\kappa(r(a)-1))y_{\tau^{-1}(a)}}\right)\nonumber\\
&&\fl\hspace{35mm} \times\prod_{j<k}^{M}\!\!\!\!\!\!\!\!\prod_{\substack{a\in\Omega_{j}\\b\in\Omega_{k}\\\tau^{-1}(a)>\tau^{-1}(b)}}\!\!\!\!\!\!\!\!\left(\frac{(q_{j}+\rmi\kappa r(a))-(q_{k}+\rmi\kappa r(b))-\rmi\kappa}{(q_{j}+\rmi\kappa r(a))-(q_{k}+\rmi\kappa r(b))+\rmi\kappa}\right)\;.
\end{eqnarray}
One notes that $\tau\in S_{n}''(\vec{n})$ is equivalent to $R(\vec{n})\circ\tau\in S_{n}'(\vec{n})$ with $R(\vec{n})$ the permutation that exchanges $a,b\in\Omega_{j}$, $j=1,\ldots,M$, if $r(a)+r(b)=n_{j}$. This implies
\begin{equation}
\widetilde{\varphi}_{\kappa}(y;M,\vec{n},q)=\overline{\varphi_{\kappa}(y;M,\vec{n},q-\rmi\kappa\vec{n}+\rmi\kappa)}\;,
\end{equation}
where $\overline{(\ldots)}$ denotes complex conjugation. Thus, one can write
\begin{eqnarray}
\label{phi[phi,phi tilde]}
&&\fl\hspace{15mm} \langle x|\rme^{-tH_{\kappa}}|y\rangle=\sum_{M=1}^{n}\frac{1}{M!(2\pi)^{M}}\sum_{\vec{n}\in D_{n,M}}\left(\int_{\mathbb{R}-\rmi\kappa\epsilon_{j}}\rmd q_{j}\right)\nonumber\\
&&\fl\hspace{20mm} \prod_{j=1}^{M}\prod_{a\in\Omega_{j}}\left(\rme^{-t(q_{j}+\rmi\kappa(r(a)-1))^{2}}\right)\varphi_{\kappa}(x;M,\vec{n},q)\overline{\varphi_{\kappa}(y;M,\vec{n},q-\rmi\kappa\vec{n}+\rmi\kappa)}\;.
\end{eqnarray}
At this point, one wants to move again the contours of integration so that the exact Bethe eigenfunctions appear. One then needs to know precisely the location of the poles of the integrand in (\ref{phi[phi,phi tilde]}). To achieve this, we prove the following lemma:
\begin{lemma}
The expression
\begin{equation}
\label{chi}
\varphi_{\kappa}(x;M,\vec{n},q)\prod_{j<k}^{M}\prod_{r=1}^{n_{j}}\prod_{s=1}^{n_{k}}\left(\frac{q_{j}-q_{k}+\rmi\kappa(r-s-1)}{q_{j}-q_{k}+\rmi\kappa(r-s)}\right)\;,
\end{equation}
with $\varphi_{\kappa}(x;M,\vec{n},q)$ defined in Eq. (\ref{phi(x)}), is holomorphic as function of the $q_{j}$'s in the whole complex plane.
\end{lemma}
\begin{proof}
One defines $\chi(x;M,\vec{n},q)$, equal to (\ref{chi}) up to a global normalization independent of $q$, by
\begin{eqnarray}
&&\fl\hspace{5mm} \chi(x;M,\vec{n},q)=\sum_{\sigma\in S_{n}'(\vec{n})}\prod_{j=1}^{M}\prod_{a\in\Omega_{j}}\left(\rme^{\rmi(q_{j}+\rmi\kappa r(a)-\rmi\kappa)x_{\sigma^{-1}(a)}}\right)\nonumber\\
&&\fl\hspace{10mm} \times\prod_{j<k}^{M}\prod_{\substack{a\in\Omega_{j}\\b\in\Omega_{k}}}\left(\frac{(q_{j}+\rmi\kappa r(a))-(q_{k}+\rmi\kappa r(b))+\rmi\kappa\,\sign(\sigma^{-1}(a)-\sigma^{-1}(b))}{(q_{j}+\rmi\kappa r(a))-(q_{k}+\rmi\kappa r(b))}\right)\;.
\end{eqnarray}
Noting that the product
\begin{equation}
\fl\hspace{20mm}\prod_{\substack{a,b\in\Omega_{j}\\a<b}}\left(\frac{(q_{j}+\rmi\kappa r(a))-(q_{j}+\rmi\kappa r(b))+\rmi\kappa\,\sign(\sigma^{-1}(a)-\sigma^{-1}(b))}{(q_{j}+\rmi\kappa r(a))-(q_{j}+\rmi\kappa r(b))}\right)\;
\end{equation}
is nonzero for all $j=1,\ldots,M$ if and only if the permutation $\sigma$ is an element of $S_{n}'(\vec{n})$, one can instead sum over all permutations $\sigma\in S_{n}$. One has
\begin{eqnarray}
&&\fl\hspace{10mm} \chi(x;M,\vec{n},q)=\prod_{j=1}^{M}n_{j}!\;\sum_{\sigma\in S_{n}}\prod_{j=1}^{M}\prod_{a\in\Omega_{j}}\left(\rme^{\rmi(q_{j}+\rmi\kappa r(a)-\rmi\kappa)x_{\sigma^{-1}(a)}}\right)\nonumber\\
&&\fl\hspace{15mm} \times\prod_{j<k}^{M}\prod_{a<b}^{n}\left(\frac{(q_{j}+\rmi\kappa r(a))-(q_{k}+\rmi\kappa r(b))+\rmi\kappa\,\sign(\sigma^{-1}(a)-\sigma^{-1}(b))}{(q_{j}+\rmi\kappa r(a))-(q_{k}+\rmi\kappa r(b))}\right)\;.
\end{eqnarray}
Then, one defines $\widetilde{\chi}(x;M,\vec{n},\xi)$ by replacing everywhere in $\chi(x;M,\vec{n},q)$ the variables $q_{j}+\rmi\kappa(r(a)-1)$, $j=1,\ldots,M$, $a\in\Omega_{j}$ by $\xi_{a}$ with the result
\begin{eqnarray}
&& \widetilde{\chi}(x;M,\vec{n},\xi)=\prod_{j=1}^{M}n_{j}!\;\sum_{\sigma\in S_{n}}\prod_{a=1}^{n}\left(\rme^{\rmi \xi_{a}x_{\sigma^{-1}(a)}}\right)\nonumber\\
&&\hspace{25mm} \times\prod_{a<b}^{n}\left(\frac{\xi_{a}-\xi_{b}+\rmi\kappa\,\sign(\sigma^{-1}(a)-\sigma^{-1}(b))}{\xi_{a}-\xi_{b}}\right)\;.
\end{eqnarray}
If $\widetilde{\chi}(x;M,\vec{n},\xi)$ is analytic in the $\xi_{a}$, then $\chi(x;M,\vec{n},q)$ will also be analytic in the $q_{j}$. One notes that $\widetilde{\chi}(x;M,\vec{n},\xi)$ has only simple poles. The term with the permutation $\sigma\in S_{n}$ of the residue at $\xi_{c}=\xi_{d}$ with $c<d$ is given by
\begin{eqnarray}
&& \Bigg(\prod_{j=1}^{M}n_{j}!\Bigg)\;\rme^{\rmi \xi_{d}(x_{\sigma^{-1}(c)}+x_{\sigma^{-1}(d)})}\prod_{\substack{a=1\\a\neq c,d}}^{n}\left(\rme^{\rmi \xi_{a}x_{\sigma^{-1}(a)}}\right)\nonumber\\
&& \times\rmi\kappa\,\sign(\sigma^{-1}(c)-\sigma^{-1}(d))\nonumber\\
&& \times\prod_{\substack{a<b\\a\neq c,d\\b\neq c,d}}^{n}\left(\frac{\xi_{a}-\xi_{b}+\rmi\kappa\,\sign(\sigma^{-1}(a)-\sigma^{-1}(b))}{\xi_{a}-\xi_{b}}\right)\nonumber\\
&& \times\prod_{\substack{b=1\\b\neq c,d}}^{n}\left(\frac{\xi_{d}-\xi_{b}+\rmi\kappa\,\sign(\sigma^{-1}(c)-\sigma^{-1}(b))}{\xi_{d}-\xi_{b}}\right)\nonumber\\
&& \times\prod_{\substack{b=1\\b\neq c,d}}^{n}\left(\frac{\xi_{d}-\xi_{b}+\rmi\kappa\,\sign(\sigma^{-1}(d)-\sigma^{-1}(b))}{\xi_{d}-\xi_{b}}\right)\;.
\end{eqnarray}
This term cancels with the term for which $\sigma^{-1}(a)$ and $\sigma^{-1}(b)$ are exchanged. Thus, $\widetilde{\chi}(x;M,\vec{n},\xi)$ and $\chi(x;M,\vec{n},q)$ are analytic functions in the $\xi_{a}$ and in the $q_{j}$. This finishes the proof of the lemma.
\end{proof}
Then, since
\begin{eqnarray}
&& \prod_{r=1}^{n_{j}}\prod_{s=1}^{n_{k}}\left(\frac{q_{j}-q_{k}+\rmi\kappa(r-s-1)}{q_{j}-q_{k}+\rmi\kappa(r-s)}\;\frac{q_{j}-q_{k}+\rmi\kappa(n_{j}-n_{k}-r+s+1)}{q_{j}-q_{k}+\rmi\kappa(n_{j}-n_{k}-r+s)}\right)\nonumber\\
&& =\frac{(q_{j}-q_{k}-\rmi\kappa n_{k})(q_{j}-q_{k}+\rmi\kappa n_{j})}{(q_{j}-q_{k})(q_{j}-q_{k}+\rmi\kappa(n_{j}-n_{k}))}\;,
\end{eqnarray}
the integrand in (\ref{phi[phi,phi tilde]}) is analytic in the $q_{j}$ except for simple poles at $q_{j}=q_{k}+\rmi\kappa n_{k}$ and $q_{k}=q_{j}+\rmi\kappa n_{j}$ for $1\leq j<k\leq M$. At this point, it is possible to set $\epsilon_{j}=0$ for all $j=1,\ldots,M$ in (\ref{phi[phi,phi tilde]}) since the poles at $q_{j}=q_{k}$ cancel. Then, the contours of integration can be moved again, from $q_{j}\in\mathbb{R}$ to $q_{j}\in\mathbb{R}+\rmi\mu_{j}$. Before moving the contours, one has (with $\kappa>0$ and $j<k$)
\begin{equation}
\fl \Im(q_{j}-q_{k}-\rmi\kappa n_{k})=-\kappa n_{k}<0\qquad\text{and}\qquad\Im(q_{k}-q_{j}-\rmi\kappa n_{j})=-\kappa n_{j}<0\;.
\end{equation}
After moving the contours, these inequalities must still be valid. One notes that if the $\mu_{j}$'s satisfy the constraint
\begin{equation}
\label{condition mu}
-n_{j}<\mu_{j}\leq0\;,
\end{equation}
both inequalities are satisfied. Shifting the contours and making then the change of variables $q_{j}\to q_{j}+\rmi\kappa\mu_{j}$ in (\ref{phi(sigma,tau)}) leads to the result (\ref{phi(kappa>0)}) of Theorem \ref{Theorem propagator 1}, while shifting the contours and making the change of variables in (\ref{phi[phi,phi tilde]}) leads to
\begin{eqnarray}
\label{phi(mu)}
&&\fl \langle x|\rme^{-tH_{\kappa}}|y\rangle=\sum_{M=1}^{n}\frac{1}{M!(2\pi)^{M}}\sum_{\vec{n}\in D_{n,M}}\int_{\mathbb{R}^{M}}\rmd q_{1}\,\ldots\,\rmd q_{M}\,\prod_{j=1}^{M}\prod_{a\in\Omega_{j}}\left(\rme^{-t(q_{j}+\rmi\kappa(\mu_{j}+r(a)-1))^{2}}\right)\nonumber\\
&&\fl\hspace{25mm} \times\varphi_{\kappa}(x;M,\vec{n},q+\rmi\kappa\vec{\mu})\;\overline{\varphi_{\kappa}(y;M,\vec{n},q+\rmi\kappa\vec{\mu}+\rmi\kappa\vec{n}-\rmi\kappa)}\;.
\end{eqnarray}
One notes that
\begin{eqnarray}
&&\fl \sum_{a\in\Omega_{j}}(q_{j}+\rmi\kappa(\mu_{j}+r(a)-1))^{2}=n_{j}q_{j}^{2}+2\rmi\kappa q_{j}\sum_{r=1}^{n_{j}}(\mu_{j}+r(a)-1)-\kappa^{2}\sum_{r=1}^{n_{j}}(\mu_{j}+r-1)^{2}\nonumber\\
&&\fl =n_{j}q_{j}^{2}+2\rmi\kappa q_{j}n_{j}\left(\mu_{j}+\frac{(n_{j}-1)}{2}\right)-\frac{\kappa^{2}}{12}(n_{j}^{3}-n_{j})-\kappa^{2}n_{j}\left(\mu_{j}+\frac{(n_{j}-1)}{2}\right)^{2}\;.
\end{eqnarray}
The choice $\mu_{j}=-(n_{j}-1)/2$ is thus the only one such that the argument of the exponential in (\ref{phi(mu)}) is a real number. It is also the only choice such that $\varphi_{\kappa}(x;M,\vec{n},q+\rmi\kappa\vec{\mu})$ and $\overline{\varphi_{\kappa}(y;M,\vec{n},q+\rmi\kappa\vec{\mu}+\rmi\kappa\vec{n}-\rmi\kappa)}$ are complex conjugates of each other. Introducing the Bethe eigenfunctions \cite{D10.2}
\begin{equation}
\label{Bethe eigenfunction}
\psi(x;M,\vec{n},q)=\varphi_{\kappa}\left(x;M,\vec{n},q-\frac{\rmi\kappa}{2}\vec{n}+\frac{\rmi\kappa}{2}\right)\;,
\end{equation}
one finds for $\mu_{j}=-(n_{j}-1)/2$
\begin{eqnarray}
\label{phi[phi]}
&& \langle x|\rme^{-tH_{\kappa}}|y\rangle=\sum_{M=1}^{n}\frac{1}{M!(2\pi)^{M}}\sum_{\vec{n}\in D_{n,M}}\int_{\mathbb{R}^{M}}\rmd q_{1}\,\ldots\,\rmd q_{M}\,\nonumber\\
&&\hspace{20mm} \psi(x;M,\vec{n},q)\;\overline{\psi(y;M,\vec{n},q)}\;\prod_{j=1}^{M}\rme^{-tn_{j}q_{j}^{2}+\frac{t\kappa^{2}}{12}(n_{j}^{3}-n_{j})}\;.
\end{eqnarray}
This is the result (\ref{phi(kappa>0,Bethe)}) of Theorem \ref{Theorem propagator 2}.


\section{The special case \texorpdfstring{$x=y=0$}{x=y=0}}
\label{Section x=y=0}
\setcounter{equation}{0}
In the special case $x=y=0$, Eq. (\ref{phi[phi]}) simplifies. One has
\begin{eqnarray}
\label{prefactor x=y=0 1}
&&\fl \sum_{\sigma\in S_{n}'(\vec{n})}\prod_{j<k}^{M}\!\!\!\!\!\!\!\!\prod_{\substack{a\in\Omega_{j}\\b\in\Omega_{k}\\\sigma^{-1}(a)>\sigma^{-1}(b)}}\!\!\!\!\!\!\!\!\left(\frac{(q_{j}+\rmi\kappa(\mu_{j}+r(a)))-(q_{k}+\rmi\kappa(\mu_{k}+r(b)))+\rmi\kappa}{(q_{j}+\rmi\kappa(\mu_{j}+r(a)))-(q_{k}+\rmi\kappa(\mu_{k}+r(b)))-\rmi\kappa}\right)\\
&&\fl =\frac{\sum\limits_{\sigma\in S_{n}'(\vec{n})}\prod\limits_{j<k}^{M}\prod\limits_{\substack{a\in\Omega_{j}\\b\in\Omega_{k}}}\Big((q_{j}-q_{k})+\rmi\kappa(\mu_{j}-\mu_{k}+r(a)-r(b)+\sign(\sigma^{-1}(a)-\sigma^{-1}(b)))\Big)}{\prod\limits_{j<k}^{M}\prod\limits_{r=1}^{n_{j}}\prod\limits_{s=1}^{n_{k}}\Big((q_{j}-q_{k})+\rmi\kappa(\mu_{j}-\mu_{k}+r-s-1)\Big)}\;.\nonumber
\end{eqnarray}
As in the previous section, one can sum instead over all permutations $\sigma\in S_{n}$ in case the product in the numerator becomes a product over all $a<b$: the permutations which do not belong to $S_{n}'(\vec{n})$ give zero contribution. We use the notation $\xi_{a}=q_{j}+\rmi\kappa(\mu_{j}+r(a))$ for $a\in\Omega_{j}$. Since the signature of the permutation $\sigma$ can be written
\begin{equation}
\sign(\sigma)=\prod_{a<b}^{n}\frac{\xi_{a}-\xi_{b}+\rmi\kappa\,\sign(\sigma^{-1}(a)-\sigma^{-1}(b))}{\xi_{\sigma(a)}-\xi_{\sigma(b)}+\rmi\kappa\,\sign(a-b)}\;,
\end{equation}
Eq. (\ref{prefactor x=y=0 1}) rewrites as
\begin{equation}
\label{prefactor x=y=0 2}
\fl\hspace{5mm} (-\rmi\kappa)^{\frac{n_{j}(n_{j}-1)}{2}}\left(\prod_{j=1}^{M}\prod_{r=1}^{n_{j}}r!\right)\frac{\sum\limits_{\sigma\in S_{n}}\sign(\sigma)\prod\limits_{a<b}^{n}\Big(\xi_{\sigma(a)}-\xi_{\sigma(b)}+\rmi\kappa\,\sign(a-b)\Big)}{\prod\limits_{j<k}^{M}\prod\limits_{r=1}^{n_{j}}\prod\limits_{s=1}^{n_{k}}\Big((q_{j}-q_{k})+\rmi\kappa(\mu_{j}-\mu_{k}+r-s-1)\Big)}\;.
\end{equation}
One has (\textit{e.g.} see \cite{PS11.1}, Lemma 1 for a proof)
\begin{lemma}
Let $f(a,b)$ be arbitrary complex coefficients. Then
\begin{equation}
\sum_{\sigma\in S_{n}}\sign(\sigma)\prod_{a<b}^{n}\big(\xi_{\sigma(a)}-\xi_{\sigma(b)}+f(a,b)\big)=n!\prod_{a<b}^{n}(\xi_{a}-\xi_{b})\;.
\end{equation}
\end{lemma}
Using this lemma, (\ref{prefactor x=y=0 2}) rewrites
\begin{eqnarray}
&&\fl\hspace{15mm} n!(-\rmi\kappa)^{-\frac{n_{j}(n_{j}-1)}{2}}\left(\prod_{j=1}^{M}\prod_{r=1}^{n_{j}}\frac{1}{r!}\right)\frac{\prod\limits_{a<b}^{n}(\xi_{a}-\xi_{b})}{\prod\limits_{j<k}^{M}\prod\limits_{r=1}^{n_{j}}\prod\limits_{s=1}^{n_{k}}\Big((q_{j}-q_{k})+\rmi\kappa(\mu_{j}-\mu_{k}+r-s-1)\Big)}\nonumber\\
&&\fl\hspace{15mm} =n!\left(\prod_{j=1}^{M}\frac{1}{n_{j}!}\right)\prod\limits_{j<k}^{M}\prod\limits_{r=1}^{n_{j}}\prod\limits_{s=1}^{n_{k}}\left(\frac{(q_{j}-q_{k})+\rmi\kappa(\mu_{j}-\mu_{k}+r-s)}{(q_{j}-q_{k})+\rmi\kappa(\mu_{j}-\mu_{k}+r-s-1)}\right)\;.
\end{eqnarray}
Then using (see \cite{D10.2}, Eq. (B.55-B.58) and Eq. (34))
\begin{eqnarray}
&& \prod_{j<k}^{M}\prod_{r=1}^{n_{j}}\prod_{s=1}^{n_{k}}\left|\frac{(q_{j}-q_{k})+\rmi\kappa(-\frac{n_{j}}{2}+\frac{n_{k}}{2}+r-s)}{(q_{j}-q_{k})+\rmi\kappa(-\frac{n_{j}}{2}+\frac{n_{k}}{2}+r-s-1)}\right|^{2}\nonumber\\
&& =\prod_{j<k}^{M}\left|\frac{(q_{j}-q_{k})-\frac{\rmi\kappa}{2}(n_{j}-n_{k})}{(q_{j}-q_{k})-\frac{\rmi\kappa}{2}(n_{j}+n_{k})}\right|^{2}\nonumber\\
&& =\det\left(\frac{\rmi\kappa n_{j}}{(q_{j}-q_{k})+\frac{\rmi\kappa}{2}(n_{j}+n_{k})}\right)_{j,k=1,\ldots,M}\;,
\end{eqnarray}
one finds
\begin{eqnarray}
\label{phi(0,0)}
&& \langle0|\rme^{-tH_{\kappa}}|0\rangle=\sum_{M=1}^{n}\frac{n!\kappa^{n}}{M!(2\pi)^{M}}\sum_{\vec{n}\in D_{n,M}}\int_{\mathbb{R}^{M}}\rmd q_{1}\,\ldots\,\rmd q_{M}\,\nonumber\\
&&\hspace{20mm} \det\left(\frac{\rme^{-tn_{j}q_{j}^{2}+\frac{t\kappa^{2}}{12}(n_{j}^{3}-n_{j})}}{-\rmi(q_{j}-q_{k})+\frac{\kappa}{2}(n_{j}+n_{k})}\right)_{j,k=1,\ldots,M}\;.
\end{eqnarray}

In the replica computations for the KPZ equation with sharp wedge initial data, Eq. (\ref{phi(0,0)}) can be used to arrive directly at the generating function of the height fluctuations, see Eq. (2.17)-(2.19) of \cite{PS11.1}.


\section{Conclusions}
Exact formulas for the transition probability of the one-dimensional asymmetric simple exclusion process, a non-equilibrium exactly solvable model, have been derived a few years ago \cite{S97.1,P03.1,TW08.1}. Subsequently the method was adapted to obtain an exact formula for the propagator of the quantum $\delta$-Bose gas with repulsive interaction \cite{TW08.2}. Here we analytically continued this formula to the case of attractive interaction.

An advantage of our approach, compared to the usual Bethe ansatz, is that the question of the completeness of the Bethe eigenfunctions can be completely avoided. In fact, such kind of exact expressions for the propagator can be used to \textit{prove} the completeness of the Bethe ansatz, at least on the infinite line. In principle, it might also be possible to use the same kind of approach for the case of periodic boundary conditions, see \textit{e.g.} \cite{P03.1,PP07.1}.

The exact expression for the propagator of the repulsive $\delta$-Bose gas, derived in \cite{TW08.2} and used here, bears some formal similarity with the coordinate Bethe ansatz, which is the original ansatz introduced by Bethe to diagonalize the Hamiltonian of Heisenberg spin chain. Since then, other descriptions of eigenstates have been developed, in particular algebraic Bethe ansatz, which makes clearer the structures underlying the quantum integrability of such type of models. It would be of interest to understand whether it is possible to write down the propagator using an approach closer to the algebraic Bethe ansatz.

\subsection*{Acknowledgments}
We thank Tomohiro Sasamoto for most instructive discussions. S.P. acknowledges the support through a DFG research project (SP181-24).


\begin{appendix}
\section{Analyticity in the coupling of the propagator of the $\delta$-Bose gas}
\label{Appendix analyticity propagator}
\setcounter{equation}{0}
In the case of a standard Schr\"odinger operator of the form $-\Delta+\lambda V$ one can use Kato's theory to establish that $\rme^{-t(-\Delta+\lambda V)}$ is analytic in $\lambda$. The $\delta$-potential corresponds to a boundary condition and we are not aware of a functional analytic argument for the holomorphic dependence on $\kappa$. Instead, we will use the Feynman-Kac representation.
\begin{proposition}
For fixed $x,y\in\mathbb{R}^{n}$, $t>0$, the function $\kappa\mapsto\langle x|\rme^{-tH_{\kappa}}|y\rangle$ is holomorphic on $\mathbb{C}$.
\end{proposition}
\begin{proof}
By the Feynman-Kac formula one has the representation
\begin{equation}
\langle x|\rme^{-tH_{\kappa}}|y\rangle=\mathbb{E}_{x,y}\left(\rme^{2\kappa X(t)}\right)p_{t}(x-y)\;.
\end{equation}
Here $p_{t}(x-y)$ is the Brownian motion transition kernel. The expectation $\mathbb{E}_{x,y}$ is over the standard Brownian bridge, $b(t)$, starting at $x$ and ending at $y$ at time $t$. Let $L_{j,k}(t)$, $j<k$, be the local time at $0$ for $\{b_{i}(s)-b_{j}(s),0\leq s\leq t\}$, \textit{i.e.}
\begin{equation}
L_{j,k}(t)=\int_{0}^{t}\rmd s\;\delta(b_{j}(s)-b_{k}(s))\;.
\end{equation}
Then
\begin{equation}
X(t)=\sum_{j<k}^{n}L_{j,k}(t)\;.
\end{equation}
Since $X(t)\geq0$, we bound as
\begin{equation}
\label{bound expX}
\mathbb{E}_{x,y}\left(\rme^{2\kappa X(t)}\right)\leq\mathbb{E}_{x,y}\left(\rme^{2|\kappa|X(t)}\right)\leq\sum_{j<k}^{n}\mathbb{E}_{x,y}\left(\rme^{|\kappa|n^{2}L_{j,k}(t)}\right)\;.
\end{equation}
The second inequality follows from the convexity of the exponential function. $L_{j,k}(t)$ is the local time at $0$ for a one-dimensional Brownian bridge starting at $x_{j}-x_{k}$ and ending at $y_{j}-y_{k}$ at time $t$. The distribution of $L_{j,k}(t)$ has a Gaussian decay at infinity. Hence for any $|\kappa|$ the right hand side in (\ref{bound expX}) is bounded. This proves analyticity on $\mathbb{C}$.
\end{proof}
\end{appendix}


\end{document}